%% file: main.tex
\pdfoutput=1
\documentclass[a4paper,UKenglish,cleveref,thm-restate]{lipics-v2021}
\input{config/include}

\input{config/icalp-lipics-document-config} 
\acknowledgements{%
The question whether the NFA acceptance problem can be solved faster
for unambiguous automata was formulated by the first author for the 2021 Dagstuhl
Seminar 21452, ``Unambiguity in Automata Theory'' (Dagstuhl Reports 11(10),
Section 3.3).
We thank Stefan Kiefer and Andrew Ryzhikov, who asked us whether
our Theorem~\ref{thm:kleafamb} resolves
the open problem from their work~\cite{kiefer_et_al:LIPIcs.STACS.2025.61,lmcs:16919}.
The answer turned out positive (Corollary~\ref{cor:01matrices}).
We are grateful
to Sarah~J.~Selkirk and Lynette van~Zijl for making this project possible
and to Michael Wehar for thought-provoking discussions.
We would also like to acknowledge the support of
the Centre for Discrete Mathematics and its Applications (DIMAP) and
the Department of Computer Science, at the University of Warwick.
}
\nolinenumbers
\hideLIPIcs
\title{Evaluating Weighted Automata:\\From the Boolean Semiring to Finite Fields}
\titlerunning{Evaluating Weighted Automata: From the Boolean Semiring to Finite Fields}
\begin{document}

\maketitle
\input{0-abstract}

\newpage
\section{Introduction}
\label{sec:introduction}
\input{1-introduction}

\section{Preliminaries}
\label{sec:preliminaries}

\input{2-preliminaries}

\section{Analysis of the dynamic programming algorithm}
\label{sec:new_algorithms_fro_the_acceptance_and_value_problems}

\input{3-new_algorithms-original}

\section{Reducing polynomially ambiguous NFA Acceptance to XNFA Acceptance}
\label{sec:accceptance_to_value}

\input{5-acceptance_to_value}

\section{Open problems}
\label{sec:discussion}
\input{6-discussion}

\bibliography{bibliography/main}
\appendix
\input{appendix}

\end{document}

%% file: config/include.tex
\input{config/package-imports}
\input{config/macros-textmode}

\input{config/macros-mathmode}
\input{config/macros-tikz}

%% file: config/package-imports.tex
\usepackage[utf8]{inputenc}
\usepackage{microtype}
\usepackage{amsfonts,amsmath,amssymb,amsthm}
\usepackage{tikz}
\usetikzlibrary{arrows.meta,positioning,calc,decorations.pathmorphing,backgrounds}
\usepackage{xparse}
\usepackage{mathtools}
\usepackage{xspace}
\usepackage{hyperref}
\usepackage{babel}
\usepackage{comment}
\usepackage{cancel}
\usepackage{nicefrac}
\usepackage{thmtools}
\usepackage{thm-restate}

\usepackage[noend]{algpseudocode}
\usepackage{mdframed}
\usepackage{algorithm}
\usepackage{bm}
\usepackage{booktabs}
\usepackage{threeparttable}
\usepackage{makecell}

%% file: config/macros-textmode.tex
\newtheorem{hypothesis}[theorem]{Hypothesis}
\newtheorem*{theorem*}{Theorem}

\newcommand{\automataClass}[1]{\textrm{#1}\xspace}
\newcommand{\NFA}{\automataClass{NFA}}
\newcommand{\XNFA}{\automataClass{XNFA}}

\newcommand{\DFA}{\automataClass{DFA}}

\newcommand{\DFAs}{\automataClass{DFAs}}

\colorlet{newgreen}{green!45!black}

\definecolor{newergreen}{rgb}{0.42,0.62,0.0}

\colorlet{fixorange}{orange!75!black}

%% file: config/macros-mathmode.tex
\newcommand{\N}{\bN}

\NewDocumentCommand{\function}{ m e{^_} O{a} m }{%
    {%
      #1%
    }%
    \IfValueT{#3}{_{#3}}%
    \IfValueT{#2}{^{#2}}%
    \if\relax\detokenize{#5}\relax
    \else
      \ifx#4a%
        \bracketInner*{#5}%
      \else
        \bracketInner[#4]{#5}%
      \fi
    \fi
}

\newcommand{\functionM}{\function}

\newcommand{\GF}[1]{\functionM{\mathbb{F}}_{#1}{}}

\renewcommand{\O}{\function{O}}

\renewcommand{\v}[1]{\vec{#1}}

\ExplSyntaxOn
\seq_set_from_clist:Nn \l_letter_seq {a,b,c,d,e,f,g,h,i,j,k,l,m,n,o,p,q,r,s,t,u,v,w,x,y,z}
\seq_set_from_clist:Nn \l_capital_letter_seq {A,B,C,D,E,F,G,H,I,J,K,L,M,N,O,P,Q,R,S,T,U,V,W,X,Y,Z}
\seq_map_inline:Nn \l_letter_seq {
    \str_if_eq:nnTF { #1 } { d }
    {
    }
    {
      \cs_new:cpn { v#1 } { \v{#1} }
    }
    \cs_new:cpn { tt#1 } { \mathtt{#1} }
}
\seq_map_inline:Nn \l_capital_letter_seq {
    \cs_new:cpn { c#1 } { \mathcal{#1} }
    \cs_new:cpn { b#1 } { \mathbb{#1} }
    \cs_new:cpn { tt#1 } { \mathtt{#1} }
}
\ExplSyntaxOff

\newcommand{\suchthat}{\suchthatSymbol\PackageWarning{Radeks Macro}{Command suchthat used outside of matching PairedDelimiter was used on input line \the\inputlineno.}}
\newcommand\suchthatSymbol[1][]{\nonscript\:#1\vert\allowbreak\nonscript\:\mathopen{}}
\DeclarePairedDelimiterX{\setInner}[1]\{\}{\renewcommand\suchthat{\suchthatSymbol[\delimsize]}#1}
\NewDocumentCommand{\set}{ O{a} m }{\ifx#1a\setInner*{#2}\else\ifx#1b{\{#2\}}\else\setInner[#1]{#2}\fi\fi}
\NewDocumentCommand{\bracket}{ O{a} m }{\ifx#1a\bracketInner*{#2}\else\bracketInner[#1]{#2}\fi}
\NewDocumentCommand{\squareBracket}{ O{a} m }{\ifx#1a\squareBracketInner*{#2}\else\squareBracketInner[#1]{#2}\fi}

\DeclarePairedDelimiter {\size}          \lvert\rvert
\DeclarePairedDelimiter {\squareBracketInner} []

\DeclarePairedDelimiter {\bracketInner}  ()

\newcommand{\isomorphism}{\mathrm{toVec}}

%% file: config/macros-tikz.tex
\colorlet{mydarkred}{red!65!black}
\colorlet{mydarkblue}{blue!60!black}

%% file: config/icalp-lipics-document-config.tex
\author{Dmitry Chistikov}{University of Warwick, Coventry, UK}{d.chistikov@warwick.ac.uk}{https://orcid.org/0000-0001-9055-918X}{Supported in part by the Royal Society International Exchanges scheme (IES\textbackslash{}R2\textbackslash{}242349) and the Engineering and Physical Sciences Research Council [EP/X03027X/1].}
\author{Radosław Pi{\'{o}}rkowski}{University of Warwick, Coventry, UK}{Radoslaw.Piorkowski@warwick.ac.uk}{https://orcid.org/0000-0002-9643-182X}{Supported by the Engineering and Physical Sciences Research Council [EP/X03027X/1].}
\author{Neha Rino}{University of Warwick, Coventry, UK}{Neha.Rino@warwick.ac.uk}{https://orcid.org/0009-0001-3273-1954}{Supported by the Feuer International Scholarship in Artificial Intelligence by University of Warwick alumnus Jonathan Feuer and by the Royal Society International Exchanges scheme (IES\textbackslash{}R2\textbackslash{}242349).}
\author{Brink {van der Merwe}}{Stellenbosch University, South Africa}{abvdm@cs.sun.ac.za}{https://orcid.org/0000-0001-5010-9934}{Supported in part by the Royal Society International Exchanges scheme (IES\textbackslash{}R2\textbackslash{}242349).}

\authorrunning{D. Chistikov, R. Pi{\'{o}}rkowski, N. Rino and B. van der Merwe} %

\Copyright{Dmitry Chistikov, Radoslaw Pi{\'{o}}rkowski, Neha Rino, Brink van der Merwe} %

\ccsdesc[500]{Theory of computation~Quantitative automata}
\keywords{Automata theory, Fine-grained Complexity, Symmetric Difference Automata} %

\category{} %

\relatedversion{} %

\EventEditors{Raghavan Komondoor, Ramprasad Saptharishi, and James Worrell}
\EventNoEds{3}
\EventLongTitle{46th IARCS Annual Conference on Foundations of Software Technology and Theoretical Computer Science (FSTTCS 2026)}
\EventShortTitle{FSTTCS 2026}
\EventAcronym{FSTTCS}
\EventYear{2026}
\EventDate{December 16--18, 2026}
\EventLocation{IIT Delhi, New Delhi, India}
\EventLogo{}
\SeriesVolume{400}
\ArticleNo{43}

%% file: 0-abstract.tex
\begin{abstract}
Acceptance is perhaps the most basic problem that one can pose
about automata.
In this paper, we compare its computational complexity
for standard nondeterministic finite automata (NFA)
and for symmetric difference automata (XOR automata, or XNFA), which accept a word if and only if the number of accepting runs
is odd.
Equivalently, NFA and XNFA are weighted automata
over the Boolean semiring and over the two-element field, respectively.
The Acceptance problem is a special case of evaluating
a weighted automaton on an input word.

We study the fine-grained complexity of these problems,
aiming to optimise the degree of the polynomial in their running-time bounds.
It is open whether there exist algorithms for these problems with running time $O((|A| \cdot |w|)^{1 - \varepsilon})$,
for any real $\varepsilon > 0$.
Here $|A|$ is the sum of the numbers of states and transitions of automaton $A$
and $|w|$ is the length of the input word $w$.

Under the assumption of polynomial ambiguity,
we provide a randomised reduction of NFA acceptance to XNFA acceptance.
This reduction opens a new algebraic route to designing faster algorithms
for polynomially ambiguous NFA Acceptance, by exploiting the finite-field structure of XNFA.
For automata of bounded ambiguity (e.g., unambiguous automata),
we show that acceptance for both NFA and XNFA
can be decided faster than in the general case.

\end{abstract}

%% file: 1-introduction.tex
Connecting and comparing algorithmic problems over different algebraic structures is a well-known challenge. Consider the complexity classes $\mathrm{NL}$ and $\oplus\mathrm{L}$  corresponding to language recognition by logarithmic-space Turing machines. $\mathrm{NL}$ machines accept a word if there is a positive number of accepting runs, while $\oplus\mathrm{L}$ machines accept if there is an odd number of accepting runs.  
It is still not known whether the inclusion $\mathrm{NL} \subseteq \oplus \mathrm{L}$ (or its converse) holds;
see~\cite{Damm90,DBLP:journals/iandc/ImmermanL95}.
A complete problem for NL is the acceptance problem for nondeterministic finite automata, that is, given an NFA $A$ and a word $w$, decide whether $A$ accepts $w$. Similarly, a $\oplus$L-complete problem is the acceptance problem for weighted automata over the two-element field~\cite{Sch61}. 

Weighted automata over the two-element field are also known in the literature as \emph{symmetric difference automata} ($\oplus$-NFA)~\cite{vdMerwe12}, or nondeterministic XOR Automata (NXA)~\cite{Vuillemin09}, or XOR-nondeterministic finite automata (XNFA)~\cite{DBLP:conf/dcfs/MaraisZ16}, or Mod $2$ Multiplicity Automata (M2MA, for regular languages over infinite words)~\cite{AngluinAFG22}.
In this paper, we use the term `XNFA'. 
A logspace reduction from NFA Acceptance to XNFA Acceptance would show that $\mathrm{NL} \subseteq \oplus \mathrm{L}$. In a similar spirit, it is not known how the time complexities of these two problems compare (as we discuss below). 
\begin{mdframed}
\makeatletter
\renewcommand\makeLineNumber{%
\hss\llap{\hb@xt@\linenumberwidth{\hspace*{-1.57em}\textsf{\tiny\LineNumber}\hss}}}
\makeatother

In this paper, we compare the computational complexity of the Acceptance problem for weighted automata over the Boolean semiring and the two-element field (NFA and XNFA, respectively).
We study the possibility of fine-grained reductions between these two problems, focusing on the time complexity.
Our results for the field $\mathrm{GF}(2)$ extend to other finite fields.

\end{mdframed}

XNFA are a variant of standard finite automata in which an
input word is accepted if and only if the number of accepting runs is odd~\cite{Zijl04a,ZijlG13}; for an extension to grammars, see~\cite{Bakinova2022}. 
Like nondeterministic and deterministic finite automata, XNFA recognise exactly the class of regular languages. 
Much like NFA, XNFA are exponentially more succinct than deterministic automata~\cite{Zijl04a};
they are incomparable with standard NFA
in terms of succinctness.
As weighted automata over a field, XNFA also enjoy properties stemming from the algebraic structure.
Their complementation requires little effort,
and the number of states
can be minimised in cubic time~\cite{Sch61,DBLP:journals/ita/CardonC80};
whilst minimising NFA is PSPACE-hard~\cite{JiangR93}.
Angluin-style learning using only polynomially many queries in the size of the minimal automaton is possible for XNFA, while this is not the case for NFA~\cite{KaznatcheevP21}.

Our work is part of a quest to understand the impact of the underlying algebraic
structure on the computational complexity of algorithmic problems.
The two-element field and the Boolean semiring are the simplest examples
of a field and an idempotent semiring, respectively.
How does this structure influence the algorithmic properties of weighted automata?

We take the `fine-grained lens':
we are interested in optimising the degrees in polynomial time
bounds, e.g., $k$ in $O(|A|^k)$ where $|A|$ is the size of the description
of the automaton~$A$. 
Fine-grained complexity aims to match upper bounds on the running time of algorithms
with lower bounds, conditional on a small number of common complexity assumptions.

\vspace*{-2.5ex}
\subparagraph*{The state of the art.}

Acceptance is perhaps the most basic problem that one can pose
about automata.
If $A$ is an NFA, the textbook algorithm for it determines, for each prefix
of the input word~$w$, which states can be reached on that prefix.
The running time is $O(|A| \cdot |w|)$, where
$|A|$ is the sum of the numbers of states and transitions of $A$
and $|w|$ is the length of~$w$.
Essentially the same algorithm can be used for XNFA acceptance,
as well as for the evaluation of weighted automata over other semirings
(cf.~\cite{DKV-handbook}).

\vspace{-2.5ex}
\subparagraph*{Is the running time optimal?}

Interestingly, it is unknown
whether substantially faster algorithms can be designed for NFA and XNFA Acceptance.
Bringmann et~al.~\cite{Bringmann24} have recently hypothesised that,
for NFA,
there is no algorithm with running time $(|A| \cdot |w|)^{1-\varepsilon}$
for any real $\varepsilon > 0$.
This \emph{NFA acceptance hypothesis} does not appear to be a consequence
of standard complexity assumptions, fine-grained or not~\cite{Bringmann24}.
The NFA acceptance hypothesis really concerns \emph{dense NFA}; i.e., those
with many more transitions than states.
Indeed, for sparse NFA (having $n$ states and $O(n)$ transitions)
a faster than $(n \cdot |w|)^{1-\varepsilon}$-time
algorithm would
contradict the well-established Orthogonal Vectors Hypothesis and
thus also the Strong Exponential Time Hypothesis (SETH)
on the complexity of Boolean satisfiability, SAT (see, e.g.,~\cite{WilliamsFG2018}).
This conditional lower bound is matched by the textbook algorithm.

The state of the art raises a question:
is XNFA Acceptance similarly hard?

\vspace{-2.5ex}
\subparagraph*{Our approach: Analysis by ambiguity.}
We consider the acceptance problems under assumptions
about the \emph{ambiguity} of automata, i.e., the possible number of runs
relative to the length of the input word.
Classes of unambiguous, finitely-ambiguous, and polynomially ambiguous
weighted automata are well-known and can be characterised in terms
of the structure of the transition graph~\cite{WeberSeidl91}.
The `ambiguity class' can be decided in polynomial time~\cite{WeberSeidl91},
and fine-grained complexity bounds have been obtained recently~\cite{DrabikDFMW25}.
Analysis by ambiguity is an established tool in the
research on weighted automata, used to obtain (un)decidability and complexity
results%
~\cite{Daviaud20,DaviaudJLMPW21,FijalkowRW22, DBLP:journals/ijac/HashiguchiIJ02, DBLP:journals/tcs/Weber94}
as well as language-theoretic properties~\cite{ChattopadhyayMMR21}.

\vspace{-2.5ex}
\subparagraph*{Our contributions.}
We show two main results.

\begin{quote}
\begin{enumerate}
\item
The value problem
for weighted automata of bounded ambiguity
can be solved faster than in time $\Theta(|\delta| \cdot |w|)$, where $\delta$ is the set of transitions.
\end{enumerate}
\end{quote}
More precisely, if
for every word $A$ has at most $k$ accepting runs on it,
the standard algorithm for the value problem has a running time of $O(k |Q| \cdot |w| + |\delta|)$,
assuming unit-cost arithmetic in the underlying algebraic structure.
An important special case is $k = O(1)$, referred to as \emph{finite ambiguity}.
While we phrase this result in terms of the value problem of weighted automata, it naturally specialises to both NFA and XNFA Acceptance. 
We highlight the case of \emph{unambiguous finite automata}
(see, e.g.,~\cite{Colcombet15}),
in which there is at most one accepting run for every $w \in \Sigma^*$.
Every unambiguous automaton can be seen both as an NFA and as an XNFA;
and every language recognised by an XNFA or NFA has an unambiguous automaton too
(albeit possibly an exponentially bigger one, by determinisation).
For $A$ unambiguous, our running time of $O(|Q| \cdot |w| + |\delta|)$  solves an open problem posed
by Kiefer and Ryzhikov~\cite{kiefer_et_al:LIPIcs.STACS.2025.61,lmcs:16919}
on the complexity of multiplying matrices from zero-one matrix monoids.

\begin{quote}
\begin{enumerate}
\setcounter{enumi}{1}
\item
There is a randomised reduction from NFA acceptance
to XNFA acceptance, under the assumption of \emph{polynomial ambiguity}.
\end{enumerate}
\end{quote}
Let $A$ be an NFA for which there exists a polynomial
$p \colon \mathbb N \to \mathbb N$ such that every word $w \in \Sigma^*$
has at most $p(|w|)$ accepting runs. 
We give a randomised algorithm 
that takes such an NFA $A$ and a word~$w$ and
produces an XNFA $B$, $|B| = \Theta(\log |w| \cdot |A|)$, such that, with
probability at least $1/4$, $A$~accepts $w$ iff at $B$ accepts
$w$. Moreover, the reduction has only one-sided error: if $A$ rejects $w$, then $B$ rejects $w$ as well.
The reduction is fast enough to yield the following conclusion: if there is
an algorithm for XNFA acceptance
running in time $O((|\delta| \cdot |w|)^{1-\varepsilon})$
for some $\varepsilon>0$, then there is a randomized algorithm for
polynomially ambiguous NFA acceptance running in time
$\O{|\delta| + |w| + (|\delta| \cdot |w|)^{1-\varepsilon}}$,
up to polylogarithmic factors.

There is nothing special about $\mathrm{GF}(2)$ that we use:
the results hold for other finite fields, if we define the acceptance
of $w$ by $A$ as evaluating to a non-zero element: $A(w) \ne 0$.

As we discuss below, prior work suggests that NFA Acceptance may already be hard for polynomially ambiguous automata. By our reduction, XNFA Acceptance cannot be easier.
The reduction raises the question of whether parity acceptance can circumvent existing complexity obstacles: can we build on this reduction
to refute the NFA Acceptance hypothesis?
We leave open the existence of faster algorithms for unrestricted (X)NFA acceptance.

\section{Related work and discussion}\label{relwork}

Fine-grained reductions showing conditional lower bounds for
automata-theoretic problems have been known for over two
decades~\cite{KarakostasLV03}.
Several years ago Potechin and Shallit~\cite{PS} showed that,
for every real $\alpha \ge 2$, if there is an algorithm
that decides in time $O(n^\alpha)$ whether an NFA with $4 n$ states accepts
a word of length $n+2$, then there is also an algorithm with the same running
time that decides if an $n$-vertex undirected graph contains a triangle.
For the latter problem, called triangle detection, the best algorithms
run in essentially the same time as matrix multiplication algorithms,
$O(n^\omega)$.

The appearance of the matrix multiplication exponent $\omega$ is
not coincidental.
As observed by Bringmann et~al.~\cite{Bringmann24}, other existing reductions
in the literature show that any (hypothetical) faster
algorithms for NFA Acceptance would need to employ fast matrix multiplication.
In a nutshell, this is established by a reduction from the $k$-Clique
problem; for brevity we refer to it as ``Bringmann et al.'s reduction''.
Bringmann et al.'s reduction, as well as Potechin and Shallit's reduction,
can be seen as evidence for
the above-mentioned NFA acceptance hypothesis.
The hypothesis itself, in effect, postulates that
fast matrix multiplication techniques cannot offer substantial
running time improvements for NFA Acceptance.

\begin{observation}
\label{obs:triangle-to-xnfa}
\hspace*{-0.2em}%
There is a randomised reduction from the triangle detection
problem to XNFA Acceptance.
\end{observation}

\begin{proof}[Proof (sketch)]
We combine two reductions from the literature, adapted suitably.
Firstly,
an existing randomised reduction from
Zero Weight Triangle to
Zero Weight Triangle Parity~\cite[Sec.~5.2]{AbboudFW20} shows
that standard triangle detection reduces to the computation
of the parity of the number of triangles.
Secondly,
the reduction from triangle detection to NFA Acceptance~\cite{PS} can
be easily changed so that the \emph{number} of triangles
coincides with the number of accepting runs for the appropriately
chosen input word.
\end{proof}

\Cref{obs:triangle-to-xnfa} above suggests that, for instance for the regime
$|Q| = \Theta(|w|)$, algorithms
for XNFA Acceptance faster than matrix multiplication
are unlikely to exist.

Intuitively, our reduction from polynomially ambiguous NFA Acceptance to
XNFA Acceptance increases the conjectured hardness level for the latter problem:
from $n^\omega$ (based on triangle detection)
to $n^3$ (based on the NFA Acceptance hypothesis on words of length $\Theta(|Q|)$).
Arguably, a catch is that we assume the NFA has polynomial ambiguity.
However, this assumption is far from trivialising the problem.

\begin{observation}
\label{obs:triangle-to-poly-amb}
In the following reductions, polynomially ambiguous NFA already
suffice for hardness:
\begin{itemize}
\item Potechin and Shallit's reduction from triangle detection~\cite{PS},
\item Bringmann et al.'s reduction from $k$-Clique~\cite{Bringmann24}.
\end{itemize}
\end{observation}

(E.g., hard instances in~\cite{PS}
have $O(|w|^3)$ accepting runs for every~$w$; and $|Q| = \Theta(|w|)$.)

Generally, even faster \emph{randomised} algorithms for the acceptance
problems are of much interest.
The state of the art for NFA Acceptance is a $|Q|^2 \cdot |w| / 2^{\Omega(\sqrt{\log |Q|})}$-time
algorithm, assuming $|\Sigma|=O(1)$:
reduce~\cite{Bringmann24}
the problem to online matrix-vector multiplication in dimension $n$,
for which the naive algorithm achieves time $O(n^3)$
but a randomised one running in time
$n^3 / 2^{\Omega(\sqrt{\log n})}$ exists~\cite{LarsenW17}.
For every $c > 0$, function $2^{c \sqrt{\log n}}$ grows faster than $(\log n)^k$
for all $k > 0$ but slower than $n^\varepsilon$ for all $\varepsilon > 0$.
No similar improvement over $n^3$ is known
for XNFA Acceptance,
and in particular, faster online matrix-vector multiplication techniques
are not available over the field of two elements~\cite{LarsenW17}.

The \emph{value} problem for weighted automata, a natural
generalisation of acceptance, is provably hard for
the min-plus semiring~\cite{BackursT17};
the argument is not known to extend to NFA. 

Related to ours are Kiefer et~al.'s results~\cite{KieferMOWW13},
deciding the equivalence of weighted automata over $(\mathbb Q, +, \times)$ in
randomised polylogarithmic parallel time.
As in the case of our reduction from NFA acceptance to XNFA acceptance,
randomness is used by one of their algorithms for polynomial identity testing, relying on the
DeMillo--Lipton--Schwartz--Zippel lemma
(a non-zero polynomial over a large enough field rarely evaluates to zero at
 a random point).
The objectives are different, however: an equivalence checker determines
if two automata assign the same weight to \emph{every} word,
whereas our reduction looks for \emph{some} accepting run in a single \NFA
on a given word,
by checking if the parity (XOR sum) of path-weights is non-zero.

Our reduction from NFA acceptance to XNFA acceptance
under the polynomial ambiguity assumption
goes against a possible intuition that `rings are easy
while semirings are hard' for weighted automata.
In support of this intuition,
contrast polynomial-time algorithms for the equivalence problem
on weighted automata
over the ring of rationals~\cite{Sch61,Tzeng92,Tzeng96,KieferMOWW13,Kiefer20}
with the PSPACE-completeness of equivalence for NFA~\cite{MeyerS72}.
Over the max-plus semiring, again with weights from~$\mathbb Q$,
equivalence is undecidable~\cite{Krob94,AlmagorBK22}.
But, by our results, the acceptance problem is no more
difficult (possibly easier) for polynomially ambiguous automata over the Boolean semiring (classical NFA) than for the field
of two elements (XNFA).

It is also interesting to compare our results with those
on the NL- and $\oplus\mathrm{L}$-complete problems of determining the value of an entry in the product of multiple matrices over the Boolean semiring and over $\mathrm{GF}(2)$,
respectively;
see the papers~\cite{Damm90,DBLP:journals/iandc/ImmermanL95}
mentioned above.
We should also mention prior results in computational complexity
which reduce nondeterministic complexity classes (such as NP and NL)
to arithmetic classes (such as $\oplus\mathrm{P}$, $\oplus\mathrm{L}$, $\mathrm{Mod}_k\mathrm{L}$, etc.).
The pioneering result of this kind is the Valiant--Vazirani theorem~\cite{VV86},
giving a randomised reduction from SAT to Unique-SAT, the version of the same
problem with the promise that there is at most one satisfying assignment.
Unique-SAT may in turn be solved by a $\oplus\mathrm{P}$ machine.
Subsequently G\'al and Wigderson proved the inclusion
$\mathrm{NL/poly} \subseteq \oplus\mathrm{L/poly}$, for the non-uniform complexity classes~\cite{GW},
which Reinhardt and Allender strengthened to the equality
$\mathrm{NL/poly} = \mathrm{UL/poly}$~\cite{RA},
where U~stands for `unambiguous' (at most one accepting computation).
Although our result is related to this line of work in nature,
our Theorem~\ref{thm:main-new} (reduction from polynomially ambiguous NFA Acceptance
to XNFA Acceptance) does not reduce nondeterminism to unambiguous nondeterminism.
In fact, our Theorem~\ref{thm:kleafamb} (and its Corollary~\ref{cor:kambXNFAacc})
show that a fine-grained reduction with this property would solve
NFA Acceptance in randomised time $O(|Q| \cdot |w|)$, up to polylogarithmic
factors, which is considered unlikely.

%% file: 2-preliminaries.tex
We begin by defining the general model of a weighted automaton.

\begin{definition}[Weighted Automaton] \label{def:WTAT} A \emph{weighted automaton} over a semiring $(\mathbb{K}, +, \cdot, 0, 1)$ is a tuple $A = (Q, \Sigma, \Delta, \lambda, \gamma)$, where:
\begin{itemize}
\item $Q$ is a finite set of states;
\item $\Sigma$ is a finite input alphabet;
\item $\Delta\colon \Sigma \rightarrow \mathbb{K}^{Q \times Q}$ assigns a transition matrix (denoted by $\Delta_{\sigma}$) to each symbol $\sigma \in \Sigma$;
\item $\lambda \in \mathbb{K}^{1 \times Q}$ is the initial weight vector;
\item $\gamma \in \mathbb{K}^{Q \times 1}$ is the final weight vector.
\end{itemize}
The value of $A$ on word $w = a_1 a_2 \dots a_n \in \Sigma^*$ is
$A(w) = \lambda \cdot \Delta_{a_1} \cdot \Delta_{a_2} \cdot \ldots \cdot \Delta_{a_n}\cdot \gamma \in \mathbb{K}$.
\end{definition}

Using this general framework, we define NFA and XNFA by specifying the underlying semiring and the interpretation of the resulting weight.

\begin{definition}[\NFA] A \emph{nondeterministic finite automaton} (\NFA) is a weighted automaton over the Boolean semiring $\mathbb{B} = (\{0,1\}, \lor, \land, 0, 1)$. 
\end{definition}

Equivalently, we can define an \NFA as a tuple $(Q, \Sigma, \delta, I, F)$, where the transition relation $\delta \subseteq Q \times \Sigma \times Q$ replaces $\Delta$, and $I$ and $F$ are subsets of $Q$ corresponding to the non-zero elements of $\lambda$ and $\gamma$, respectively. Whenever there is a path in the underlying labelled graph $(Q, \Sigma, \delta)$ from a state $q$ to a state $q'$, we say $q'$ is reachable from $q$ and $q$ co-reachable from~$q'$. An automaton is \emph{trim} when every state is reachable from $I$ and co-reachable from $F$.  

An \emph{accepting run} of $A$ on a word $w$ is a path $\rho$ in the underlying labelled graph $(Q, \Sigma, \delta)$ that starts at an initial state $i \in I$, ends in a final state $f \in F$, and whose edge labels yield $w$ when concatenated. 
A word $w$ is accepted if $A(w) = 1$, which corresponds to the existence of at least one accepting run of $A$ on $w$.

\begin{definition}[\XNFA] A \emph{symmetric difference automaton} (\XNFA) is a weighted automaton over the finite field of two elements $\mathrm{GF}(2) = (\{0,1\}, \oplus, \cdot, 0, 1)$.  
\end{definition}

Similarly to NFA, we can define an \XNFA as a tuple $(Q, \Sigma, \delta, I, F)$, and an accepting run of $A$ on a word $w$ is a $w$-labelled path from $I$ to $F$. An \emph{XNFA $A$ accepts a word} $w$ if $A(w) = 1$, or equivalently, if the number of accepting runs of $A$ on $w$ is odd.

\begin{definition}[Ambiguity~\cite{WeberSeidl91}] Let $A$ be an \NFA or \XNFA. 
For a word $w \in \Sigma^*$, let $d_A(w)$ denote the number of accepting runs of $A$ on $w$.
We say $A$ is \emph{$k$-ambiguous} if $d_A(w) \le k$ for all $w$; \emph{finitely ambiguous} if $A$ is $k$-ambiguous for some constant $k \ge 1$;
\emph{polynomially ambiguous} if $d_A(w)$ is $\O{p(\size{w})}$ for some polynomial $p$; and \emph{exponentially ambiguous} otherwise.
\end{definition}

If an \NFA (or \XNFA) is not polynomially ambiguous, then 
$d_A(w_i)$ is $2^{\Omega(|w_i|)}$ for some infinite sequence of words $(w_i)_{i\in\N}$~\cite{WeberSeidl91}.

Given a weighted automaton $A$, we call a sequence of states and letters $q_0 \xrightarrow[]{\sigma_1} q_1  \dots \xrightarrow[]{\sigma_n} q_n$ an \emph{accepting run of $A$} whenever $\lambda[q_0] \neq 0$, $\gamma[q_n] \neq 0$, and for all $i \in \{1, \dots, n\}$, $\Delta_{\sigma_i}[q_{i-1}][q_i] \neq 0$. Note that an accepting path may have overall weight zero, if the semiring $\mathbb{K}$ has zero divisors. We can define the notions of $k$-ambiguous, polynomially ambiguous, and exponentially ambiguous weighted automata analogously.

When \NFA or \XNFA transition matrices are sparse (i.e., most entries are zero), representing them as full $Q \times Q$ matrices is inefficient. In an \NFA or \XNFA, representing the set $\set{q \suchthat (p, \sigma,q)\in\delta}$, for all $p \in Q$ and  $\sigma \in \Sigma$ as a singly-linked list is standard and convenient for sparse automata.

%% file: 3-new_algorithms-original.tex
The Acceptance problems for NFA and XNFA are special cases of the value problem for weighted automata: given a weighted automaton $A$ and a word $w$, compute $A(w)$ (see \Cref{def:WTAT}). In this section, we show that, when the automaton~$A$ is $k$-ambiguous, the value problem can be decided faster than in the general case.

We describe a standard dynamic programming algorithm that computes the value of $A$ on~$w$.
We prove that, when $A$ has a bounded number of runs to each final state on~$w$, this algorithm does not make as many transition traversals as in the worst case. As a consequence, under the assumption of $k$-ambiguity, this algorithm is guaranteed to run faster than the standard bound of $O((|Q|+ |\delta|)\, |w|)$. 

Let $I$ be the initial set, i.e., the set of all states that have non-zero entries in $\lambda$, and $F$ be the final set, i.e., the set of all states with non-zero entries in $\gamma$. 

\newcommand*{\code}{\texttt}
\newcommand*{\myif}{\textbf{if}\ }
\newcommand*{\mythen}{\textbf{then}\ }
\newcommand*{\myelse}{\textbf{else}\ }
\newcommand*{\myor}{\textbf{or}\ }
\newcommand*{\mycontinue}{\textbf{continue}\ }
\newcommand*{\mybreak}{\textbf{break}\ }
\newcommand*{\myreturn}{\textbf{return}\ }
\newcommand*{\mytrue}{\textbf{true}\ }
\newcommand*{\myfalse}{\textbf{false}\ }
\newcommand*{\myaccept}{\textbf{accept}\ }
\newcommand*{\myreject}{\textbf{reject}\ }
\newcommand*{\mydo}{\textbf{do}}
\begin{algorithm}
        \begin{algorithmic}[1]
            \Require Weighted automaton $A$ over a semiring $\mathbb{K}$ and word $w \in \Sigma^*$
            \Ensure The value of $A$ on $w$. 
            \State trim $A$ \Comment{ensure every state is reachable from $I$ and co-reachable from $F$}
            \State $\mathsf{prev}[q] \gets \lambda[q],\ \mathsf{curr}[q] \gets 0 \textbf{ for each } q \in Q$
             \For{$i = 1$ to $|w|$}
                \State $\mathsf{curr}[q] \gets 0 \textbf{ for each } q \in Q$
                \State \textbf{for each} $p \textbf{ if }\mathsf{prev}[p] \neq 0$ \mydo
                    \State \hspace*{\algorithmicindent}\textbf{for each} $q \in \delta(p, w[i])$ \mydo
                        \State \hspace*{\algorithmicindent}\hspace*{\algorithmicindent}$\mathsf{curr}[q] \gets \mathsf{curr}[q] + \mathsf{prev}[p]\cdot \Delta_{w[i]}[p][q]$ 
                        
                        \Comment{here $+$ and $ \cdot$ are semiring operations in $\mathbb{K}$}
                          \State $\mathsf{prev} \gets \mathsf{curr}$                    
            \EndFor
            \State \myreturn $\mathsf{curr}  \cdot \gamma$
        \end{algorithmic}
\caption{Standard dynamic programming algorithm for evaluation of weighted automata.}
\label{alg:k-amb-algo}
\end{algorithm}

\begin{theorem}\label{thm:kleafamb}
    \Cref{alg:k-amb-algo} computes the value of a weighted automaton $A$ on a word $w$. If, for all words $w' \in \Sigma^*$, $A$ has, for every final state $f \in F$,  at most $k \ge 1$ runs from $I$  to $f$ on $w'$, then \Cref{alg:k-amb-algo} uses at most $2 k \, |w| \, |Q|$ semiring operations and additionally $O(k \, |w| \, |Q|  + |\delta|)$ time.
\end{theorem}

\begin{remark}
     Even when $A$ is $k$-ambiguous and trimmed, a single state may have an outdegree greater than $k|\Sigma|$; $k$-ambiguity bounds the number of accepting runs (on each input string), not the local branching of the transition graph.
\end{remark}

\begin{proof}  
We claim that during every iteration of the outermost \textbf{for} loop, for each state $q \in Q$, the entry $\textsf{curr}[q]$ is assigned (in lines 5--7) at most $k$ times.

For all $f \in F$, let $A$ have at most $k$ runs from $I$ to $f$ on $w$. Consider iteration $i \leq \size{w}$. Suppose there is a state $q \in Q$ such that $\textsf{curr}[q]$ is assigned more than $k$ times during this iteration. This means that there are at least $k+1$ paths from the initial set $I$ to $q$ labelled by the same prefix $w[1, i]$. By line 1, since $A$ is trim, there must exist a word $v \in \Sigma^*$ that labels a path from $q$ to some $f \in F$. This would yield $k+1$ runs on the word $w[1,i]v$, from $I$ to $f$, contradicting our assumption.

To ensure the given running time bound,
note that each of the sets $\delta(p, \sigma)$ for every state $p \in Q$ and every letter $\sigma \in \Sigma$ can be implemented as a singly-linked list. 
Trimming the automaton in line~1 takes $O(|Q| + |\delta|)$ time. 
Every iteration of the outermost \textbf{for} loop performs at most $k\, |Q|$ assignments to $\textsf{curr}$, and lines~4 and 8 perform an additional $|Q|$ assignments to \textsf{curr} and \textsf{prev} each. Overall, there are exactly $|w|$ iterations.

Hence,  the algorithm  runs in  $O(k\, \size{w}\, \size{Q} + \size{\delta})$ time.
\end{proof}

\subsection{Finitely ambiguous NFA Acceptance}
Note that a $k$-ambiguous automaton always has the property that, for every word $w$ and final state $f \in F$, there are at most $k$ runs from $I$ to $f$ on $w$. In terms of the usual degree of ambiguity of the automaton, \Cref{thm:kleafamb} gives us the following bound: 

\begin{corollary}\label{cor:kambweightedvalue}
  Given a $k$-ambiguous weighted automaton $A$ and a word $w$, \Cref{alg:k-amb-algo} computes the value $A(w)$ using at most $2 k \, |w| \, |Q|$ semiring operations and additionally $O(k \, |w| \, |Q|  + |\delta|)$ time.
\end{corollary}

As an immediate consequence, the acceptance problem for finitely ambiguous NFA and XNFA is decidable faster than in time $O(|\delta||w|)$ too.
\begin{corollary}\label{cor:kambXNFAacc}
    Given a $k$-ambiguous NFA (or XNFA) $A$ and a word $w$, \Cref{alg:k-amb-algo} decides NFA Acceptance (respectively, XNFA Acceptance) in $O( k|w| \, |Q| + |\delta|)$ time. 
\end{corollary}
Whenever $k \, |Q| = o(|\delta|)$, this improves upon the naive NFA acceptance bound of $O(|w|\, |\delta|)$. By setting $k=1$,  \Cref{alg:k-amb-algo} decides if an unambiguous automaton $A$ accepts $w$ in $O( |w| \, |Q| + |\delta|)$ time. 

\subsection{Zero--one matrix products can be computed in quadratic time}
A consequence of our tighter bound for unambiguous NFA Acceptance is the following bound on the complexity of 0--1 matrix multiplication if the product is known to be a 0--1 matrix too. This resolves an open problem posed by Kiefer and Ryzhikov~\cite{kiefer_et_al:LIPIcs.STACS.2025.61,lmcs:16919}: 

\begin{corollary}\label{cor:01matrices}
    Let $P$ and $R$ be 0--1 matrices of size $n \times n$ such that the
product $P\cdot R$ is also a 0--1 matrix. Then $P\cdot R$ can be computed from $P$ and $R$ in
time $O(n^2)$.
\end{corollary} 

\begin{proof} Form the automaton $A$ over the alphabet $\{a\}$ on states $\{q_1,\dots , q_n\} \times \{1, 2, 3\}$ with the transitions defined as $(q_i, 1) \xrightarrow[]{a} (q_j, 2)$ if and only if  $P[i,j] = 1$ and $(q_i, 2) \xrightarrow{a} (q_j, 3)$ if and only if $R[i,j] = 1$. 
The $i$th row of $P\cdot R$ has $1$ in the $j$th component precisely when $(q_i, 1)$ reaches $(q_j, 3)$ on reading $aa$. 

First, delete any state that is not co-reachable from $Q \times \{3\}$. This can be done in $O(n^2)$ time. 
Then run \Cref{alg:k-amb-algo} (without any trimming) $n$ times, for the $i$th iteration setting $(q_i, 1)$ as the unique initial state. Each run of \Cref{alg:k-amb-algo} computes the $i$th row of $P\cdot R$. 
Since $P\cdot R$ is a 0--1 matrix, $aa$ has at most one run from $(q_i, 1)$ to  $(q_j,3)$ for all $i, j$. Since $A$ is unary and tripartite, all other words $w' \in \{a\}^*$ have no runs, i.e., they are rejected. By \Cref{thm:kleafamb}, the $i$th row of $P\cdot R$ can be computed by \Cref{alg:k-amb-algo}  in time $O(|Q|\cdot |w|) = O(n)$, so $P\cdot R$ can be computed in time $O(n^2)$. 
\end{proof}

\subsection{Unambiguous NFA Acceptance versus Unique NFA Acceptance} We were unable to complement our $O(|Q||w|)$ time upper bound for the acceptance problem for unambiguous NFA with a matching lower bound. Instead, we were able to identify a related, more general variant of the unambiguous NFA Acceptance problem that we call the Unique NFA Acceptance problem. While the best known upper bound for Unique NFA Acceptance remains $|Q|^2 \cdot |w| / 2^{\Omega(\sqrt{\log |Q|})}$ (\cite{Bringmann24}, via reduction from NFA Acceptance to the online matrix-vector multiplication problem \cite{LarsenW17}), we provide an SETH lower bound that rules out $O((|Q||w|)^{1-\varepsilon})$ time algorithms for Unique NFA Acceptance, for all $\varepsilon > 0$. 

\begin{definition}[{Unique Automaton Value}]\label{def:unique}
{ The \emph{Unique Automaton Value} problem takes as input a weighted automaton $A$ over a semiring $\mathbb{K}$ and a word $w \in \Sigma^*$, \emph{with the promise that $A$ has at most one accepting run on $w$}, and outputs the value $A(w)$. Its restrictions to the Boolean semiring and to $\mathrm{GF}(2)$ are the \emph{Unique \NFA{} Acceptance} and \emph{Unique \XNFA{} Acceptance} problems, whose output is whether $A$ accepts $w$.}
\end{definition}

Note that unambiguous NFA Acceptance is a special case of Unique NFA Acceptance, since all unambiguous NFA Acceptance instances are Unique NFA Acceptance instances.
The converse is not true, because the Unique problems place no restriction
on the number of accepting runs on words other than~$w$.
Unlike for its unambiguous counterpart, it is not known if there is an $O((|\delta||w|)^{1-\varepsilon})$-time algorithm for the Unique Automaton Value problem.

Our lower bound for the Unique NFA Acceptance problem arises from the conjectured fine-grained hardness of SAT. Recall the Strong Exponential Time Hypothesis (SETH), a well-known fine-grained complexity hypothesis.
For $k \ge 2$ (unrelated to the ambiguity bound above), $k$-SAT is the decision problem in which the input is an $n$-variable Boolean CNF formula $\varphi$ with exactly $k$ variables per clause and the output is whether $\varphi$ is satisfiable.

\begin{hypothesis}[SETH = Strong Exponential Time Hypothesis~\cite{DBLP:journals/jcss/ImpagliazzoP01, DBLP:conf/iwpec/CalabroIP09}]
\mbox{}\\
Let $s_k = \inf \{\,\delta \mid \text{there exists an $O(2^{\delta n})$-time (randomized) algorithm for $k\text{-SAT}$}\,\}$. Then
\begin{equation*}
\lim_{k \to \infty} s_k = 1.
\end{equation*}
\end{hypothesis}

\begin{theorem} \label{thm:unique} For all $\varepsilon > 0$, if there is an $O((|Q||w|)^{1-\varepsilon})$ time algorithm for Unique NFA Acceptance, then SETH is false. 
\end{theorem}

The main idea behind the proof of this theorem is as follows: we echo Backurs and Indyk's proof \cite[Theorem 1]{Backurs16} that the general NFA Acceptance problem has no $O((|Q||w|)^{1-\varepsilon})$-time algorithm for any $\varepsilon > 0$ unless SETH is false. (In fact, their result is a lower bound for the closely related problem of regular expression pattern matching, but we use a standard translation from regular expressions to NFA.) They first rely on the fact that the $k$-SAT problem reduces to the $2$-orthogonal vectors problem ($2$-OV) \cite{WilliamsOV}. Building on this fact, Backurs and Indyk reduce $2$-OV to NFA Acceptance. 

We produce the same chain of reductions between the unique versions of these problems, i.e., we show that a suitably defined Unique-SAT reduces to Unique $2$-OV, which in turn reduces to Unique NFA Acceptance. These reductions are fine-grained, i.e., they allow us to conclude that for every $\varepsilon>0$ if there is a $O((|Q||w|)^{1-\varepsilon})$-time algorithm for Unique NFA Acceptance, then there is an $O(2^{(1-\delta)n})$-time algorithm for Unique-SAT for some $\delta = \delta(\varepsilon) > 0$. Since Theorem 1 in~\cite{DBLP:journals/jcss/CalabroIKP08} shows that an $O(2^{(1-\delta)n})$-time algorithm for Unique-SAT would imply an $O(2^{(1 -\gamma)n})$-time algorithm for SAT for some $\gamma = \gamma(\delta) > 0$, this would falsify SETH.

%% file: 5-acceptance_to_value.tex
Our main contribution in this section is the following reduction, which shows that for polynomially ambiguous NFA, there is a randomized (one-sided error), near-linear time reduction from the NFA Acceptance problem to the XNFA Acceptance problem.

\begin{theorem}
\label{thm:main-new}
Let $A$ be a polynomially ambiguous \NFA with $|Q|$ states, $|\delta|$ transitions, and let $w\in\Sigma^+$.
Let $k >  \lceil \log_2 |w| \rceil $.
There exists an algorithm that produces an XNFA $B$ from $A$, such that:
\begin{enumerate}
\item If $A$ rejects $w$, then $B$ rejects $w$ with probability $1$.
\item If $A$ accepts $w$, then $B$ accepts $w$ with probability at least $\frac{1}{2} \cdot (1-\frac{|w|}{2^k}) \ge \frac{1}{4}$.
\item The number of states and transitions in $B$ is $k\cdot |Q|$ and at most $k^2\cdot |\delta|$, respectively.
\item The algorithm runs in time $O(|\delta|\cdot \mathrm{poly}(k))$.
\end{enumerate}
\end{theorem}

Hence if one could obtain an $O((m |w|)^{1-\varepsilon})$-time algorithm for XNFA Acceptance for some $\varepsilon > 0$, then using our reduction we get a $O((m |w|)^{1-\varepsilon})$-time randomised algorithm for NFA Acceptance for polynomially ambiguous NFA. 

\begin{corollary}\label{cor:reduction-finegrained}
Suppose that, for some real $\varepsilon > 0$, there is an algorithm that
decides XNFA Acceptance for every XNFA $C$ with transition relation
$\delta^{(C)}$ and every input word $u$ in time
$\O{\bigl(\left|\delta^{(C)}\right|\cdot|u|\bigr)^{1-\varepsilon}}$.
Then the acceptance problem for a polynomially ambiguous NFA $A$ with
transition relation $\delta$ and input word $w$ admits a randomised algorithm
with one-sided error of at most $\tfrac{1}{2}$, running in time
$\widetilde{O}\!\left(|\delta| + |w| + \bigl(|\delta|\cdot|w|\bigr)^{1-\varepsilon}\right)$,
where $\widetilde{O}$ hides factors polylogarithmic in $|w|$.
\end{corollary}

While this does not give a faster deterministic algorithm for NFA Acceptance in general, note that the NFA Acceptance Hypothesis \cite{Bringmann24} rules out faster randomised algorithms as well, and that currently all the known hard instances of NFA Acceptance are polynomially ambiguous themselves (see our Observation 2). 

Our main technical tool is a reduction given by the following lemma.

\begin{lemma}
\label{lem:main}
Let $A$ be a polynomially ambiguous \NFA with $|Q|$ states, $|\delta|$ transitions, and let $w\in\Sigma^+$.
Let $k >  \lceil  \log_2 |w| \rceil $.
There exists an algorithm that produces $k$ XNFA $A_i$ from $A$, all identical except for their sets of accepting states, such that:
\begin{enumerate}
\item If $A$ rejects $w$, then $A_i$ rejects $w$ for all $i$, with probability $1$.
\item If $A$ accepts $w$, then $A_{i}$ accepts $w$ for at least one $i$, with probability at least $1-\frac{|w|}{2^k}$.
\item The number of states and transitions in each $A_i$ is $k\cdot |Q|$ and at most $k^2\cdot |\delta|$, respectively.
\item The algorithm runs in time $O(|\delta|\cdot \mathrm{poly}(k))$.
\end{enumerate}
\end{lemma}

The reduction has two steps. We first encode \NFA acceptance of $w$ as the question of whether a weighted automaton over a field $\mathrm{GF}(2^k)$ computes a non-zero value on the input $w$. We then translate this question into \XNFA acceptance. To achieve these transformations, we apply the Schwartz-Zippel Lemma for polynomial rings over fields (\cref{lem:SZ}).

\subsection{Polynomial ambiguity, the Schwartz-Zippel lemma, and $k$-XNFAs}

In this subsection we build up to the proof of \Cref{lem:main}.

We are given an \NFA $A$, and a word $w$, as an arbitrary instance of NFA Acceptance. We first observe that we can assume without loss of generality that $A$ has a single initial state, as otherwise we could construct in $O(|Q| + |\delta| + |w|)$-time an \NFA $A'$ and a word $w'$ such that $A'$ accepts $w'$ if and only if $A$ accepts $w$. (Let $a \in \Sigma$, and set $w' = aw$. To construct $A'$, take $A$, add a new initial state $q$, add $a$-transitions from $q$ to all of the original initial states of $A$, and make all of the original initial states of $A$ non-initial.) 

Let $\GF{}$ be a field.
By $\GF{}[x_1, \ldots, x_m]$ we denote the ring of (multivariate) polynomials
with variables $x_1, \ldots, x_m$ and coefficients from $\GF{}$.

In order to apply the Schwartz-Zippel Lemma, we first turn an arbitrary \NFA $A$ with a single initial state into a weighted automaton, denoted $A_{\GF{}[\vec x]}$, over a suitable ring of polynomials. 
In more detail, take $\vec x = (x_1, \ldots, x_{|\delta|})$, one variable $x_t$ per transition~$t$.
For a field~$\GF{}$, the weighted automaton $A_{\GF{}[\vec x]}$ over the ring
$\GF{}[\vec x]$ is obtained from $A$ by using $x_t$ as the weight of transition~$t$,
for all $t \in Q \times \Sigma \times Q$.
By $A_{\GF{}[\vec x]}(w)$ we denote the multivariate polynomial (over the variables $\vec x$) obtained when evaluating this weighted automaton on $w$.  

We next establish a relationship between the acceptance of a word $w$ by $A$ and the corresponding value $A_{\GF{}[\vec x]}(w)$, under the assumption that $A$ is polynomially ambiguous.

For a word $u = t_1 t_2 \cdots t_l \in T^*$, over the alphabet $T$, its \emph{Parikh image} $P_u\colon T \to \N$ is the function that returns the number of times each symbol in $T$ appears in $u$. For a run $\pi \in \delta^*$, Parikh image $P_\pi\colon \delta \to \N$ yields the number of times $\pi$ uses each transition in $\delta$. In the next result, without loss of generality (at the cost of adding one state to $Q$), we assume that the automaton has exactly one initial state.

\begin{lemma}\label{lem:parikh}
    Let $A = (Q, \Sigma, \delta, \{q_0\}, F)$ be an \NFA.
    Assume there exists $w\in\Sigma^*$ that admits two different accepting runs with the same Parikh image.
    Then $A$ has exponential ambiguity.
\end{lemma}

\begin{figure}
    \centering
    \begin{tikzpicture}[
  >=Stealth,
  node distance=1.4cm,
  state/.style={circle,draw,minimum size=6mm,inner sep=1pt,fill=white,font=\small},
  runA/.style={->,thick,mydarkred},                 %
  runB/.style={->,thick,mydarkblue},                  %
  runAw/.style={runA,decorate,decoration={snake,amplitude=.3mm,segment length=2.5mm,post length=1.5mm}},
  runBw/.style={runB,decorate,decoration={snake,amplitude=.3mm,segment length=2.5mm,post length=1.5mm}},
  sharedw/.style={->,thick,black!75,decorate,decoration={snake,amplitude=.3mm,segment length=2.5mm,post length=1.5mm}},
  sharedR/.style={mydarkred, thick, dash pattern=on 4.5pt off 5.5pt},
  sharedB/.style={mydarkblue, thick, dash pattern=on 4.5pt off 5.5pt, dash phase=5pt},
  sharedRw/.style={sharedR, decorate, decoration={snake,amplitude=.3mm,segment length=2.5mm,post length=1.5mm}},
  sharedBw/.style={sharedB, decorate, decoration={snake,amplitude=.3mm,segment length=2.5mm,post length=1.5mm}},
  retw/.style={->,thick,decorate,decoration={snake,amplitude=.3mm,segment length=2.5mm,post length=1.5mm}},
  lbl/.style={font=\small,inner sep=1pt,text=black},
  scale=0.87
]

\begin{scope}[on background layer]
  \draw[gray!55,rounded corners=4pt,fill=gray!5] (3,-2.3) rectangle (12,2.3);
\end{scope}
\node[gray!40!black,font=\large] at (3.55,1.8) {$G_c$};

\node[state] (q0) at (0,0) {$q_0$};
\coordinate (qip) at (3,0);
\node[state] (qb) at (6,0)  {$q_b$};
\node[state] (r)  at (8,1.5)  {$r$};
\node[state] (s)  at (8,-1.5) {$s$};
\node[state] (qe) at (12,0)  {$q_e$};
\node[state,double] (qf) at (15,0) {$q_f$};

\draw[sharedBw,shorten >=3pt] (q0) -- (qip);
\draw[sharedRw,->] (q0) -- %
node[lbl,above=3pt] {${\color{mydarkred}\pi[1,i'-1]}$}
node[lbl,below=3pt] {$\color{mydarkblue}\pi'[1,i'-1]$}
                   (qip);
\draw[sharedBw,shorten >=3pt] (qip) -- (qb);
\draw[sharedRw,->] (qip) -- %
node[lbl,above=3pt] {${\color{mydarkred}\pi[i',i-1]}$}
node[lbl,below=3pt] {$\color{mydarkblue}\pi'[i',i-1]$}
                   (qb);

\draw[runA] (qb) -- node[lbl,above=3pt,pos=0.47] {$a$}
            node[lbl,below right=0pt,text=mydarkred] {${\pi[i]}$} (r);
\draw[sharedB,shorten >=3pt] (qb) -- (s);
\draw[sharedR,->] (qb) edge
                    node[lbl,auto,above=4pt,text=black] {$a$}
                  node[lbl,auto,below=0pt,pos=0.7,anchor=north east] {${\color{mydarkred}\pi[j]\color{black}, \color{mydarkblue}\pi'[i]}$} (s);

\draw[runAw] (r) to[bend right=45] node[lbl,pos=0.6,above=11pt,anchor=south]
  {$\color{mydarkred}\pi[i+1,j-1]$} (qb);

\draw[runAw,->] (s) to[out=0,in=270-30] node[lbl,pos=0.4,anchor=north west,below=4pt] {$\color{mydarkred}{\pi[j+1,j'-1]}$} (qe);
\draw[runBw] (s) to[out=60,in=270-60] node[lbl,pos=0.55,below=2pt] {$\color{mydarkblue}{\pi'[i+1,j'-1]}$} (qe);

\draw[retw] (qe) -- node[lbl,pos=.5,above=3pt,text=black] {$\pi_l$} (qb);

\draw[runAw,->] (qe) to[out=-30,in=-150] node[lbl,below=3pt] {$\color{mydarkred}\pi[j',k]$} (qf);
\draw[runBw,->] (qe) to[out=30,in=150] 
node[lbl,above=3pt] {$\color{mydarkblue}\pi'[j',k]$}
                 (qf);

\end{tikzpicture}
    \caption{Runs $\pi$ and $\pi'$ with equal Parikh images appearing in the proof of \cref{lem:parikh}.}
    \label{fig:lem-parikh-picture}
\end{figure}

\begin{proof}
Let $w=a_1 \cdots a_k$, then $w[i] = a_i$ for $i \in \set{1, 2, \ldots, k}$, and $w[i,j] = a_i a_{i+1}\cdots a_j$.
Let $u \in \delta^*$ be a sequence of transitions. Let
    $G_A(u)$ be a subgraph of the automaton graph containing only edges from $u$.
    Formally, $G_A(u) = (Q^\prime, E)$, where $E = \{t\mid u[i] = t\text{ for some $i$}\}$, and $Q^\prime$ is the minimal subset of $Q$ incident with edges in $E$.

    We rely on the following structural characterisation of exponential ambiguity due to Weber and Seidl~\cite{WeberSeidl91}.
    A state $q \in Q$ is \emph{useful} if it is reachable from the initial state and co-reachable from a final state.
    \begin{claim}[Weber, Seidl]
        \label{fact:exponentially-amb}
        An NFA $A$ is exponentially ambiguous if and only if
        there exist a useful state $q \in Q$ and two different runs of $A$ from $q$ to $q$ labelled by the same word $v \in \Sigma^*$.
    \end{claim}

      Assume $w \in \Sigma^*$ has two different accepting runs $\pi, \pi' \in \delta^*$ with the same Parikh image.
    We show the existence of a useful state $q \in Q$  and a word such that there are two different runs from $q$ to $q$ labelled by this word.
    The following construction is depicted in \cref{fig:lem-parikh-picture}.
     
    Let $i$ be the first position such that $\pi[i] \neq \pi'[i]$.
    Therefore, $\pi[i] = (q_b, a, r)$ and $\pi'[i] = (q_b, a, s)$ for some $q_b,r,s \in Q$, $r \neq s$, and $a \in \Sigma$. Moreover, $q_b$ lies on an accepting run $\pi$, hence it is useful.
    
    By equality of Parikh images of $\pi$ and $\pi'$, the transition $\pi'[i]$ appears in $\pi$ at some $j > i$, i.e., $\pi[j] = \pi'[i]$.
    Therefore, $G_A(\pi[i,j-1])$ is strongly connected.
    Choose minimal $i' \leq i$ and maximal $j' \geq j$ such that $G_c = G_A(\pi[i', j'-1])$ is strongly connected.
    By strong connectivity, $\pi$ does not use the edges of $G_c$ outside the infix $\pi[i', j'-1]$.
 By the equality of Parikh images, $G_c$ is also a strongly connected component in $G_A(\pi')$.
    By strong connectivity, $\pi'$ uses the edges of $G_c$ on some infix $\pi'[i'', j''-1]$.
    But $q_b$ is a vertex of $G_c$ and is visited by both $\pi$ and $\pi'$, and the prefixes of $\pi$ and $\pi'$ before $i$ are equal, so
    both runs enter $G_c$ at the same time, i.e., $i' = i''$.
    By equality of Parikh images, the intervals have equal lengths: $j' - i' = \sum_{t \in E} P_\pi(t) = \sum_{t \in E} P_{\pi'}(t) = j'' - i'$, thus $j'' = j'$, i.e., the runs leave $G_c$ at the same time.
    Therefore, both $\pi$ and $\pi'$ use the edges of $G_c$ exactly in the infixes $\pi[i', j'-1]$ and $\pi'[i', j'-1]$, both labelled by $w[i', j'-1]$.
    
 By the equality of Parikh images, the runs exit $G_c$ via the same edge $(q_e, \cdot, \cdot)$, so they both visit the same state $q_e$ as the last vertex of $G_c$: $\pi[j'-1] = (\cdot, \cdot, q_e)$ and $\pi'[j'-1] = (\cdot, \cdot, q_e)$.
    By strong connectivity, there exists some $w_l \in \Sigma^*$ such that there is a run $\pi_l$ from $q_e$ to $q_b$ labelled by $w_l$ in $G_c$.
    Then $\pi[i,j'-1]\pi_l$ and $\pi'[i,j'-1]\pi_l$ are two different closed walks from $q_b$ to $q_b$, labelled with the same word $w[i,j'-1]w_l$.
    Therefore, by \cref{fact:exponentially-amb}, $A$ is exponentially ambiguous.
\end{proof}

A polynomial $p \in \GF{}[x_1, \ldots, x_m]$ is \emph{non-zero}
if it contains at least one monomial (term), and
it represents a \emph{non-zero function} if there exists some
$a = (a_1, \ldots, a_m) \in \GF{}^m$ such that $p(a) \ne 0$.
A non-zero multivariate polynomial of degree~$d$ over~$\GF{}$
is guaranteed to represent a non-zero function if $|\GF{}| > d$;
see, e.g.,~\cite[Lem.~3.1]{Moshkovitz10}.

\begin{observation}
\label{lem:syntax}
A polynomially ambiguous \NFA $A$ accepts $w$ precisely when the polynomial $A_{\GF{}[\vec x]}(w)$ is non-zero:
from \cref{lem:parikh} we see that different runs for a given input word evaluate to different monomials.
\end{observation}

Recall that if $A$ is exponentially ambiguous, then it contains a useful state that has two cycles on the same input (see~\cite{WeberSeidl91}). These cycles can be permuted. This justifies the `polynomially ambiguous' assumption in \cref{lem:syntax}.

Next, we recall the Schwartz-Zippel lemma in the form required in this section.
We write $\Pr[\ldots]$ to denote the probability of an event.

\begin{lemma}[Schwartz-Zippel~\cite{Schwartz80}]\label{lem:SZ}
Let $\GF{}$ be a finite field and $p\in \GF{}[x_1,\dots,x_m]$ a polynomial of total degree $d$ representing a non-zero function.
If each $r_i$ is chosen independently and uniformly from $\GF{}$, then
$
\Pr\bigl[p(r_1,\dots,r_m)=0\bigr] \;\le\; \frac{d}{\size{\GF{}}}
$.
\end{lemma}

\begin{remark}\label{rem:repeat}
We may repeat the random assignment of values to variables $r$ times, until $\left(\nicefrac{d}{|\GF{}|}\right)^r$ becomes smaller than some required error (assuming $d<|\GF{}|$).
\end{remark}

By abuse of notation, we denote by $A_{\GF{}}$ the weighted automaton over $\GF{}$, obtained by taking $A_{\GF{}[\vec x]}$ and replacing each variable weight $x_t$ with a random value from $\GF{}$. Next, we consider the relationship between whether $A$ accepts $w$ or not and the value of $A_{\GF{}}(w)$.

\begin{lemma}\label{lem:NFA_SZ}
Assume that \NFA $A$ is polynomially ambiguous, $\GF{}$ is a finite field, and $|\GF{}| > |w|$. If $A$ rejects $w$, then $A_{\GF{}}(w)=0$, and if $A$ accepts $w$, then $\Pr[A_{\GF{}}(w)=0] \le \frac{|w|}{|\GF{}|}$.
\end{lemma}

\begin{proof}
 By construction, the weighted automaton $A_{\GF{}}$ sums the weights (from $\GF{}$) of all accepting runs of $A$ on input $w$. If $A$ has no accepting runs on $w$, then this sum is empty and $A_{\GF{}}(w)=0$ deterministically.

Next, assume that $A$ has at least one accepting run on $w$. Because $A$ is polynomially ambiguous, $A_{\GF{}[\vec x]}(w)$ is non-zero by \cref{lem:syntax}. Since $A_{\GF{}[\vec x]}(w)$ is a $|w|$-degree polynomial, and we are evaluating over the field $\GF{}$ with $|\GF{}| > |w|$, we know that $A_{\GF{}[\vec x]}(w)$ represents a non-zero function~\cite[Lem.~3.1]{Moshkovitz10}.
The result now follows from the Schwartz-Zippel Lemma, noting that $A_{\GF{}}(w)$ is equal to the polynomial $A_{\GF{}[\vec x]}(w)$ evaluated by replacing $x_t$ with the value from $F$ assigned to $t$ when defining $A_{\GF{}}$. 
\end{proof}

In our setting, the field $\GF{}$ will be a finite field of cardinality $2^k$, for some value of $k$: Galois field $\mathrm{GF}(2^k)$. Next, we consider how to go from a weighted automaton  over $\mathrm{GF}(2^k)$ to XNFA.
Fix a basis of $\mathrm{GF}(2^k)$ as a $k$-dimensional vector space over $\mathrm{GF}(2)$. Multiplication by any fixed element from $\mathrm{GF}(2^k)$ is a $\mathrm{GF}(2)$-linear map and can therefore be represented by a $k\times k$ matrix over $\mathrm{GF}(2)$.

To extend the parity acceptance concept to $k$-dimensional field representations, we define a generalised weighted automaton model, which we refer to as  $k$-\XNFA. Although the definition can be used on any semiring, we only use it on $\mathrm{GF}(2)$, and thus we refer to these weighted automata as $k$-XNFA.

\begin{definition}\label{def:wa}
A \emph{$k$-\XNFA} $A$ is a weighted automaton over $\mathrm{GF}(2)$, but with a final weight matrix $\gamma \in \mathrm{GF}(2)^{Q \times k}$. The word \( w = a_1 a_2 \ldots a_l \in \Sigma^* \) is accepted if
$\lambda \cdot \Delta_{a_1} \cdot \Delta_{a_2} \cdots \Delta_{a_l} \cdot \gamma$ is a non-zero vector in $\mathrm{GF}(2)^{1\times k}$.
\end{definition}

Every $1$-\XNFA is, trivially, an \XNFA. We refer to the acceptance condition of a $k$-\XNFA as the \emph{generalised parity acceptance condition}.
We show how to simulate a weighted automaton over $\mathrm{GF}(2^k)$ by a $k$-\XNFA, with a small overhead on the state and transition count.

\begin{restatable}{lemma}{GFsimulation}
    \label{lem:simulate-general}
    Let $\GF{} = \mathrm{GF}(2^k)$, and $A_F$ be a weighted automaton over $\GF{}$ with $|Q|$ states and $|\delta|$ transitions, and 
with all entries in the initial and final vectors being 0 or 1.
One may construct a $k$-\XNFA $B$ with $k\cdot |Q|$ states and at most $k^2\cdot |\delta|$ transitions 
such that, for every~$w$, the row vector $B(w)$ equals the coordinate vector of $A_{\GF{}}(w)$ with respect to a chosen basis of~$\GF{}$.
This construction can be done in $\O{|\delta| \cdot \mathrm{poly}(k)}$ time.
\end{restatable}

The construction is based on the isomorphism $\mathrm{GF}(2^k) \cong \mathrm{GF}(2)^k$ as vector spaces over $\mathrm{GF}(2)$. The $k$-\XNFA $B$ simulates the evaluation of $A_{\GF{}}$ by replacing each state of a weighted automaton over $\mathrm{GF}(2^k)$ with a vector of $k$ states. Each transition with weight $r$ is replaced by the corresponding linear update (matrix multiplication over $\mathrm{GF}(2)$) implementing $g\mapsto r\cdot g$ in $\mathrm{GF}(2^k)$.
The details of the construction are relegated to the appendix.

A $k$-\XNFA can be split into $k$ ordinary XNFA with essentially no overhead:

\begin{restatable}[$k$-\XNFA decomposition]{lemma}{XNFAdecomposition}
\label{lem:decompose}
For a $k$-\XNFA $A = (Q, \Sigma, \Delta, \lambda, \gamma)$, consider $k$ XNFA
$A_i = (Q, \Sigma, \Delta, \lambda, \gamma_i), i = 1, \ldots, k$,
where $\gamma_i$ denotes the $i$th column of $\gamma$. Then $A$ accepts $w$ iff at least one of the XNFA $A_i$ accepts $w$.
\end{restatable}

Details of the proofs are provided in the appendix.

\subsection{Sketch of the proof of \Cref{lem:main}}

Set $\GF{} = \mathrm{GF}(2^k)$ so that $|\GF{}| = 2^k > |w|$. \Cref{alg:nfa-to-xnfa} (page~\pageref{alg:nfa-to-xnfa}) has three steps.

\smallskip\noindent\textbf{\textsf{Step 1 (Randomisation).}}
Sample independently and uniformly a weight from $\GF{}$ for each of the $|\delta|$ transitions of $A$, and let $A_{\GF{}}$ be the resulting weighted automaton over $\GF{}$ (as defined before \cref{lem:NFA_SZ}).

\smallskip\noindent\textbf{\textsf{Step 2 (Simulation).}}
Apply \cref{lem:simulate-general} to $A_{\GF{}}$ to obtain a $k$-\XNFA $B$ with $k\cdot |Q|$ states and at most $k^2\cdot |\delta|$ transitions such that $B$ accepts $w$ iff $A_{\GF{}}(w)\ne 0$.

\smallskip\noindent\textbf{\textsf{Step 3 (Decomposition).}}
Apply \cref{lem:decompose} to $B$ to obtain $k$ XNFA $A_1, \ldots, A_k$. By the lemma, they share the state set and transitions of $B$, differ only in their accepting states, and $B$ accepts $w$ iff at least one $A_i$ accepts $w$.

\smallskip
We now verify each part of the statement of Lemma~\ref{lem:main}.

\emph{Parts~(1) and (2): correctness.}
Since $|\GF{}| = 2^k > |w|$, \cref{lem:NFA_SZ} applies. If $A$ rejects $w$, then $A_{\GF{}}(w) = 0$ deterministically; by \cref{lem:simulate-general}, $B$ rejects $w$; by \cref{lem:decompose}, every $A_i$ rejects $w$. If $A$ accepts $w$, then $A_{\GF{}}(w)\ne 0$ with probability at least $1 - |w|/|\GF{}| = 1 - |w|/2^k$; on this event, \cref{lem:simulate-general} gives that $B$ accepts $w$, and \cref{lem:decompose} gives that at least one $A_i$ accepts $w$.

\emph{Part~(3): size.}
By \cref{lem:simulate-general}, $B$ has $k\cdot |Q|$ states and at most $k^2\cdot |\delta|$ transitions. By \cref{lem:decompose}, each $A_i$ inherits the same state set and transitions.

\emph{Part~(4): running time.}
Step~1 samples $|\delta|$ elements of $\GF{}$, each of $k$ bits, in time $O(|\delta|\cdot k)$. Step~2 runs in time $O(|\delta|\cdot \mathrm{poly}(k))$ by \cref{lem:simulate-general}. Step~3 only duplicates pointers to $B$ while restricting the final weight vector, so it runs in time linear in the output size. The total running time is $O(|\delta|\cdot \mathrm{poly}(k))$.

Note that the construction of the Galois field $\mathrm{GF}(2^k)$ fits within our running time bound, because a $k$th degree irreducible polynomial can be found in time
polynomial in $k$~\cite{Shoup}.

Directly from the second part of \Cref{lem:main}, we obtain the following corollary:

\begin{corollary}\label{cor:half}
With the assumptions as in \cref{lem:main}, suppose
$k > \left\lceil \log_2 |w|\right\rceil$.
If $A$ accepts $w$, then at least one of $A_1, \dots , A_k$ accepts
$w$ with probability at least $\nicefrac{1}{2}$.
\end{corollary}

\begin{algorithm}
\caption{Reduction from polynomially ambiguous NFA Acceptance to XNFA Acceptance.}
\label{alg:nfa-to-xnfa}
\begin{algorithmic}[1]
\Require polynomially ambiguous NFA $A=(Q,\Sigma,\delta,\lambda,\gamma)$ with $0/1$ initial and
         final vectors, word $w\in\Sigma^{+}$, integer $k>\lceil \log_2|w|\rceil $.
\Ensure XNFA $A_1,\dots,A_k$, identical except for
        their final vectors, such that: if $A$ rejects $w$ then every $A_i$ rejects $w$,
        and if $A$ accepts $w$ then with probability
        $\ge 1 - |w|/2^{k}$ at least one $A_i$ accepts $w$.
\State $\triangleright$ \emph{Step 0: fix a representation of $\mathbb{F}=\mathrm{GF}(2^{k})$}
\State compute an irreducible $f(y)\in\mathrm{GF}(2)[y]$ of degree $k$, so that
       $\mathbb{F}\cong\mathrm{GF}(2)[y]/\langle f\rangle$ with basis $1,y,\dots,y^{k-1}$
\State let $\mathrm{toVec}\colon\mathbb{F}\to\mathrm{GF}(2)^{k}$ be the coordinate map
\State for $g\in\mathbb{F}$
       let $M_g\in\mathrm{GF}(2)^{k\times k}$ be the matrix of $h\mapsto g\cdot h$ in this basis
\State $\triangleright$ \emph{Step 1: randomisation}
\For{\textbf{each} transition $t=(p,\sigma,q)\in\delta$}
    \State sample $r_t\in\mathbb{F}$ independently and uniformly at random
\EndFor
\State $\triangleright$ \emph{Step 2: simulation -- build $k$-XNFA $B=(Q',\Sigma,\Delta',\lambda',\gamma')$}
\State $Q' \gets \{\,(q,1),\dots,(q,k) : q \in Q\,\}$
       \Comment{$k|Q|$ states}
       \State $S_q \gets \{(q,1),\dots,(q,k)\}$
\For{\textbf{each} $\sigma\in\Sigma$}
    \State $\Delta'_{\sigma}\gets \mathbf{0}\in\mathrm{GF}(2)^{k|Q|\times k|Q|}$
    \For{\textbf{each} transition $t=(p,\sigma,q)\in\delta$}
        \State set the $S_p \times S_q$ block of $\Delta'_{\sigma}$ to $M_{r_t}^{\mathsf{T}}$
    \EndFor
\EndFor
\State $\lambda'\gets \mathbf{0}\in\mathrm{GF}(2)^{1\times k|Q|}$
\For{\textbf{each} $q\in Q$ with $\lambda_q=1$}
    \State set the $S_q$ block of $\lambda'$ to $\mathrm{toVec}(1)=(1,0,\dots,0)$
\EndFor
\State $\gamma' \gets \mathbf{0}\in\mathrm{GF}(2)^{k|Q|\times k}$
\For{\textbf{each} $q\in Q$ with $\gamma_q=1$}
    \State set the $q$th $k\times k$ block of $\gamma'$ to $I_k$
\EndFor
\State $\triangleright$ \emph{Step 3: decomposition into $k$ ordinary XNFA}
\For{$i\gets 1$ \textbf{to} $k$}
    \State $A_i\gets(Q',\Sigma,\Delta',\lambda',\gamma'^{(i)})$, where $\gamma'^{(i)}$ is the $i$th column of $\gamma'$
\EndFor
\State \Return $A_1,\dots,A_k$
\end{algorithmic}
\end{algorithm}

\subsection{Sketch of the proof of \Cref{thm:main-new}}

Instead of constructing $k$ different XNFAs (as in Lemma~\ref{lem:main}) by applying \cref{lem:decompose}, we could instead construct a single XNFA as follows.

Let $c$ be a uniformly random $k \times 1$ dimensional vector over $\GF{2}$. We take the $k$-XNFA $B$ produced by applying \cref{lem:simulate-general}, and replace its final weight matrix $\gamma$ with the final weight vector $\gamma \cdot c$. Let the resulting XNFA be called $B'$. If $A$ rejects $w$, then $B'$ rejects $w$. If $A$ accepts $w$, then $B'$ accepts $w$ with probability $\nicefrac{1}{2} (1 - \nicefrac{|w|}{2^k})$, which in our case is at least $\nicefrac{1}{4}$. 

We think of this as random compression of $k$-XNFA into a single XNFA.

\subsection{Sketch of the proof of \Cref{cor:reduction-finegrained}}

\begin{proof}
We may assume $w \in \Sigma^{+}$, as the case $w = \varepsilon$ is decided
directly from the initial and final states. Set $k = \lceil \log_2 |w| \rceil + 1$
so that $k = \Theta(\log |w|)$ and $|w|/2^{k} \le \tfrac{1}{2}$.
By Theorem~\ref{lem:main} we construct, in time $O(|\delta|\cdot\mathrm{poly}(k))$,
the $k$ XNFA $A_1,\dots,A_k$, each with at most $k^{2}\cdot|\delta|$ transitions,
such that (i) if $A$ rejects $w$ then every $A_i$ rejects $w$, and (ii) if $A$
accepts $w$ then at least one $A_i$ accepts $w$ with probability at least
$\tfrac{1}{2}$ (Corollary~\ref{cor:half}). Running the assumed XNFA Acceptance
algorithm on each $A_i$ with input $w$ therefore decides the acceptance of $A$ on $w$
with one-sided error at most $\tfrac{1}{2}$.

For the running time, each of the $k$ calls costs
$O\!\left((k^{2}|\delta|\cdot|w|)^{1-\varepsilon}\right)$, writing the $k$ input
words costs $O(|w|\,k)$, and the construction costs $O(|\delta|\cdot\mathrm{poly}(k))$.
The total is
\[O\!\left(|\delta|\cdot\mathrm{poly}(k) + |w|\,k
        + k\cdot\bigl(k^{2}|\delta|\cdot|w|\bigr)^{1-\varepsilon}\right).\]
Since $k = \Theta(\log |w|)$, the factors $\mathrm{poly}(k)$ and
$k\cdot(k^{2})^{1-\varepsilon}$ are polylogarithmic in $|w|$ and are absorbed into
$\widetilde{O}$, giving
$\widetilde{O}\!\left(|\delta| + |w| + (|\delta|\cdot|w|)^{1-\varepsilon}\right)$.
The general error bound $\eta\in(0,1)$ is treated in
the appendix.
\end{proof}

\begin{remark}
    The above reduction from polynomially ambiguous NFA Acceptance to XNFA Acceptance can be easily generalised to reducing polynomially ambiguous NFA Acceptance to the value problem for weighted automata over other finite fields. Since the strength of our reduction is that even the simplest field, GF(2), can efficiently simulate polynomially ambiguous NFA Acceptance, we state our results in terms of XNFA Acceptance. 
\end{remark}

%% file: 6-discussion.tex
We leave it open
whether our randomised reduction can be made deterministic
(error-less).
Finding a fine-grained reduction in the other direction is also of
interest. We do not know of any problems that are provably hard for XNFA, but not for NFA.
Minimal XNFA and minimal \DFAs are linked:
if we determinise a minimal \XNFA, then we obtain the
corresponding minimal \DFA~\cite{vdMerwe12}.

Our work raises the question of whether parity acceptance can circumvent
existing complexity obstacles.
Is it possible to use our reduction to refute the NFA Acceptance
hypothesis~\cite{Bringmann24}, or
is \XNFA acceptance equally hard in a fine-grained sense?
In the case $|\Sigma| = 1$, both NFA and XNFA acceptance
can be solved in time $O(|Q|^\omega \log|w|)$, see~\cite{PS}.

Without the assumption of polynomial ambiguity, NFA Acceptance still
reduces to the non-zeroness of a polynomial, but the
variables must no longer be commuting.
The obstacle to overcome (or avoid) here is the high running time
(on this polynomial) of state-of-the-art methods
for non-commutative polynomial identity testing.

%% file: appendix.tex
\section{SETH lower bound for Unique NFA Acceptance}\label{app:UniqueApp}
In this appendix, we will provide the full proof of \Cref{thm:unique}, restated below.

\begin{theorem*}For all $\varepsilon > 0$, if there is an $O((|Q||w|)^{1-\varepsilon})$ time algorithm for Unique NFA Acceptance, then SETH is false. 
\end{theorem*}
\begin{proof}

The structure of this proof is as follows: we echo Backurs and Indyk's proof \cite[Theorem 1]{Backurs16} that the general NFA Acceptance problem has no $O((|Q||w|)^{1-\varepsilon})$-time algorithm for any $\varepsilon > 0$ unless SETH is false. They first rely on the fact that the $k$-SAT problem reduces to the $2$-orthogonal vectors problem ($2$-OV) \cite{WilliamsOV}. Building on this fact, Backurs and Indyk reduce $2$-OV to regular expression pattern matching, a problem that is closely related to NFA Acceptance. Note that the classic Thompson construction (see, e.g., Lemma 1.55 in~\cite{sipser2013introduction})  converts regular expressions to language-equivalent $\varepsilon$-NFA,  and performing $\varepsilon$-removal on this resulting $\varepsilon$-NFA yields a language-equivalent NFA. By chaining Backurs and Indyk's reduction with this conversion from regular expressions to NFA, we obtain a reduction from $2$-OV to NFA Acceptance, which we will describe in greater detail below.  

In our proof, we produce the same chain of reductions between the unique solution versions of these problems, i.e., we show that Unique-SAT reduces to Unique $2$-OV, which in turn reduces to Unique NFA Acceptance. These reductions will be fine-grained, i.e., they will allow us to conclude that for any $\varepsilon>0$ if there is a $O((|Q||w|)^{1-\varepsilon})$-time algorithm for Unique NFA Acceptance, then there is an $O(2^{(1-\delta)n})$-time algorithm for Unique-SAT. Since Theorem 1 in~\cite{DBLP:journals/jcss/CalabroIKP08} shows that an $O(2^{(1-\delta)n})$-time algorithm for Unique-SAT would imply an $O(2^{n-\gamma})$-time algorithm for SAT, this would imply that SETH is false.

\paragraph*{Step 1: Reducing Unique-SAT to Unique $2$-OV}
The two main problems we will need for this reduction are: 
\begin{enumerate}
    \item The $2$-Orthogonal Vectors ($2$-OV) problem:  The decision problem that takes as input two sets $A, B \subseteq \{0,1\}^d$ of $|A| = |B| = n$ vectors of dimension $d = c \log n$, and outputs whether  $A$ contains a vector orthogonal to a vector in $B$.
    \item The \textsc{$k$-SAT} problem: The decision problem that given a $k$-CNF Boolean formula, outputs whether the formula is satisfiable or not. 
\end{enumerate}

The unique versions of these problems are: \begin{enumerate}
    \item \textsc{Unique $k$-SAT} is $k$-SAT with the additional promise that the input $k$-CNF formula has at most one satisfying assignment. 
    \item \textsc{Unique $2$-OV} is $2$-OV with the additional promise that there is at most one pair of vectors $(u, v)$ such that $u \in A$, $v \in B$, and $u\cdot v = 0$. 
\end{enumerate}

We note that the classic fine-grained reduction from $k$-SAT to $2$-OV~\cite{WilliamsOV} preserves the uniqueness of solutions. The reduction has two steps, sparsification and encoding orthogonal vectors, described in brief below:

\begin{description}
    \item[Sparsification] For any $\varepsilon > 0$, suppose we are given a $k$-CNF formula $\Phi$ over $n$ variables (with possibly $O(n^k)$ clauses). By Corollary 1 of Impagliazzo, Paturi, and Zane~\cite{DBLP:journals/jcss/ImpagliazzoPZ01}, in time $2^{\varepsilon n} \cdot \text{poly}(n)$, we can construct $2^{\varepsilon n}$ $k$-CNF formulae $\Psi_i$ over $n$ variables that have at most $cn$ clauses each (where $c$ is a constant that only depends on $\varepsilon$ and $k$), such that $\Phi = \bigvee_{i = 1}^{2^{\varepsilon n}} \Psi_i$.  

    \item[OV encoding] Now, we study the orthogonal vector encoding from~\cite{WilliamsOV} of each of these sparsified $k$-CNF formulae. Let $\Psi$ be an $n$ variable $k$-CNF formula with $m = O(n)$  clauses. The sets $A$ and $B$ consist of $N = 2^{n/2}$ vectors, each of dimension $m$.  $A$ (respectively, $B$) is a set of vectors $\{u_{b} \text{ (respectively, $v_{b}$)}  \mid b \in \{0,1\}^m\}$ such that  $u_{b}[i] \text{ (respectively, $v_{b}[i]$) } $  is 0 $\iff$  setting the first (respectively, last) $n/2$ variables to $b$ satisfies the $i$th clause. 
    
    $\Psi$ has a satisfying assignment if and only if there is at least one pair of vectors $(u_{b_1},v_{b_2})$ such that $b_1 b_2$ satisfies $\Psi$, which means that at the $i$th coordinate at least one of $(u_{b_1}[i],v_{b_2}[i])$ is zero, hence $u \cdot v = 0$.  
\end{description} 

We first note that, if the $k$-CNF instance $\Phi$ is promised to have had at most one satisfying assignment, then since $\Phi = \bigvee_{i = 1}^{2^{\varepsilon n}} \Psi_i$, each of the $\Psi_i$ produced by sparsification must also have at most one satisfying assignment.
Secondly, if some sparsified $k$-CNF formula $\Psi$ has at most one satisfying assignment, then since each assignment corresponds to a pair of orthogonal vectors, there is at most one orthogonal vector pair, thereby producing a Unique $2$-OV instance.

\paragraph*{Step 2: Reducing Unique $2$-OV to Unique NFA Acceptance}

We will now show that the reduction by Backurs and Indyk~\cite[Thm.~1]{Backurs16} also preserves uniqueness of solutions, allowing us to reduce \textsc{Unique $2$-OV} to Unique NFA acceptance. To do so, we first briefly describe a slight variant of the reduction by Backurs and Indyk, adapted to be a reduction from $2$-OV to NFA Acceptance (rather than a reduction to regular expression pattern matching). 

Suppose we are given two sets $A = \{a_i \mid i \leq N\}$ and $B =\{b_i \mid i \leq N\}$ of  $(c \log N)$-dimensional vectors for some constant $c$, such that there exists at most one pair $i, j \leq N$ with $a_i \cdot b_j = 0$. The reduction constructs an NFA $M$ and a word $w$ such that $w \in L(M)$ if and only if $A$ and $B$ contain an orthogonal pair of vectors, and $w$ has at most one accepting run on $M$. 

Let $v$ be a word $v \in \{0, 1\}^*$, where $v = v_1 \cdots  v_{|v|}$, with $v_i \in \{0,1\}$ for $i \in \{1, \dots |v|\}$.  We define $\textsf{Chain}(v)$ to be an (incomplete) DFA consisting of $|v|+1$ states $q_0, \dots q_{|v|}$. The initial state of $\textsf{Chain}(v)$ is $q_0$, and its sole final state is $q_{|v|}$. 

The transition function of $\textsf{Chain}(v)$ is defined as follows: for all $i \in \{1, \dots |v|\}$, whenever $v_i = 0$, we have that $q_{i-1} \xrightarrow[]{0} q_i$  and $q_{i-1} \xrightarrow[]{1} q_i$ are both transitions in $\textsf{Chain}(v)$. Otherwise, if  $v_i = 1$, then we have only that $q_{i-1} \xrightarrow[]{0} q_i$, i.e., there is no transition labelled by $1$ between $q_{i-1}$ and $q_i$. There are no other transitions. The language of $\textsf{Chain}(v)$ consists of the set of vectors orthogonal to $v$. 

For each  $b \in B$, we construct $\textsf{Chain}(b)$. The NFA $M$ is the automaton obtained by gluing each of the initial states of all of the $\textsf{Chain}(b)$ DFAs together into a single initial state, and gluing each of the final states of the $\textsf{Chain}(b)$ DFAs into a single final state. 
To the initial state and the final state of $M$, we add self loops on $0, \#$ and $1$. 

The word $w$ is the marked concatenation of the vectors in $A$, i.e.,  $w = a_1 \#  a_2 \dots \# a_N$. A word $w$ has an accepting run in $M$ if and only if one of the vectors $a_i$ is orthogonal to one of the vectors $b_j$. 

We note that if there is exactly one orthogonal pair in $A, B$, then there is exactly one accepting path for the word $w$. The number of states in $M$ and the length of the word $w$ are both at most $c N \log N$. Hence, Unique $k$-SAT reduces to Unique $2$-OV, which in turn reduces to Unique NFA acceptance.

\paragraph*{Step 3: Concluding SETH-hardness}

The last ingredient we require to conclude SETH-hardness is Theorem 1 in~\cite{DBLP:journals/jcss/CalabroIKP08}. 

Let $\sigma_k = \inf \{\delta \mid \text{ there exists an } O(2^{\delta n}) \text{ time (randomized) algorithm for Unique } k\text{-SAT}\}$. By Theorem 1 in~\cite{DBLP:journals/jcss/CalabroIKP08}, it is known that $\lim_{k \to \infty} s_k = \lim_{k \to \infty} \sigma_k$. Hence, for all $\varepsilon > 0$, if there is a $2^{(1-\varepsilon)n}$-time algorithm for Unique $k$-SAT, then $\lim_{k \to \infty} s_k \neq 1$, i.e., SETH is false.

    Now, we are ready to complete the proof. 
    
    Suppose there exists an $\varepsilon > 0$ such that there is an $O((|Q||w|)^{1-\varepsilon})$-time algorithm for Unique NFA Acceptance. Then, by the reduction described in step 2, there is an $O(N^{2(1-\varepsilon)})$-time algorithm for Unique $2$-OV. By the reduction described in step 1, this yields an $O(2^{(1-\delta)n})$ algorithm for Unique $k$-SAT, for some $\delta > 0$. By Theorem 1 in~\cite{DBLP:journals/jcss/CalabroIKP08}, this means that for some $\gamma> 0$, $\lim_{k \to \infty} s_k = 1 - \gamma \neq 1$ i.e.\ SETH is false. 
\end{proof}

\section{Reduction from NFA Acceptance to XNFA Acceptance}

Here we detail the construction of a $k$-XNFA that simulates the evaluation of a weighted automaton over $\mathrm{GF(2^k)}$. 

\GFsimulation*

\begin{definition}[Coordinate simulation]
\label{def:coordinate-construction}
    Given a weighted automaton $A_{\GF{}}$ over $\mathrm{GF}(2^k)$, we construct a $k$-XNFA $B$ in the following way:
\begin{enumerate}
    \item \emph{Representation.} Find an irreducible polynomial $f(y)\in\mathrm{GF}(2)[y]$ of degree $k$, so that the fields $\mathrm{GF}(2^k)$ and $ \mathrm{GF}(2)[y]/\langle f\rangle$ are isomorphic, and use the basis $1, y, \ldots, y^{k-1}$. 
    Note that $\mathrm{GF}(2)[y]/\langle f\rangle$ and  $\mathrm{GF}(2)^k$ are isomorphic vector spaces, and let $\isomorphism\colon \mathrm{GF}(2^k) \to \mathrm{GF}(2)^k$ be the composition of these isomorphisms.
    \item \emph{Matrix computation.} For each of the (at most) $|\delta|$ transitions of $A_{\GF{}}$, with weight $g\in\mathrm{GF}(2^k)$, compute the $k\times k$ matrix $M_g$ representing $h\mapsto g\cdot h$ in the basis $1, y, \ldots, y^{k-1}$. The $i$th column of $M_g$ is obtained by multiplying $g$ by $y^{i-1}$ to produce a polynomial of degree at most $2k-2$, and then reducing the result modulo $f$. 
    \item \emph{Assembly.} 
    Label each original transition in $A_F$ carrying weight $g$ with the
transpose $M_g^{\mathsf T}$ of the matrix computed in step (2) (so that, for a row vector $\mathrm{toVec}(h)\cdot M_g^{\mathsf T} = \mathrm{toVec}(g\cdot h)$, it is consistent with the row vector acceptance of Definition \ref{def:wa}). Replicate each state $q$ into $k$ field-component states $S_q = \{q_1,\ldots,q_k\}$. 

    The initial vector of $A_{\GF{}}$ is expanded coordinatewise, yielding
$\lambda_B \in \mathrm{GF}(2)^{1\times k\cdot|Q|}$. The final matrix
$\gamma_B \in \mathrm{GF}(2)^{k|Q|\times k}$ is defined blockwise: since the final
weights are $0$ or $1$, the $q$th $k\times k$ block of $\gamma_B$ is
$I_k$ if $\gamma_q=1$, and $0_{k\times k}$ otherwise.
\end{enumerate}
This produces a $k$-\XNFA with $k\cdot |Q|$ states; since each $M_r$ contributes at most $k^2$ non-zero entries, the total number of transitions is at most $k^2\cdot |\delta|$.
\end{definition}

\begin{lemma}[Coordinate simulation]\label{lem:simulate}
Let $\GF{} = \mathrm{GF}(2^k)$.  
Given a weighted automaton $A_{\GF{}}$ over $\GF{}$ with $|Q|$ states and $|\delta|$ transitions, and 
with all entries in the initial and final vectors being 0 or 1, \cref{def:coordinate-construction} produces a $k$-\XNFA $B$
such that, for every $w$
, the row vector

$B(w)$ equals the coordinate vector of $A_{\GF{}}(w)$ with respect to a chosen basis of $\GF{}$. 
\end{lemma}
In particular, $B$ accepts $w$ (under the generalised parity acceptance condition) if and only if $A_{\GF{}}(w) \ne 0$.
\begin{proof}
Fix a basis $b_1,\dots,b_k$ of $F$ over $\mathrm{GF}(2)$. The proof is by induction on the length of the prefix of the input word $w = a_1 a_2 \cdots a_\ell$. The invariant is that $\mathrm{toVec}\big((\lambda_A \prod_{j\le i}\Delta_A(a_j))[q]\big) = (\lambda_B \prod_{j\le i}\Delta_B(a_j))[S_q]$, where the $[S_q]$ notation selects only the entries corresponding to states~$S_q$.

For each $r \in \GF{}$ the map $g \mapsto r\cdot g$ is $\mathrm{GF}(2)$-linear and is represented on column coordinate vectors by a $k\times k$ matrix $M_r$; since $B$ propagates a row vector, the corresponding block is $M_r^{\mathsf T}$, which gives
$\mathrm{toVec}(c_p)\cdot M_r^{\mathsf T} = \mathrm{toVec}(r\cdot c_p)$ and so preserves the invariant. The final multiplication by $\gamma_B$ then sums exactly the coordinate vectors of the field values accumulated at the final states of $A_F$.

The evaluation of $B$ on $w$ therefore returns exactly the coordinate vector of $A_F(w)$ in the basis $b_1,\dots,b_k$. This vector is non-zero in $\mathrm{GF}(2)^k$ if and only if $A_F(w) \ne 0$ in $F$, which is the acceptance condition of $B$.
\end{proof}

\begin{lemma}[Efficiency of the simulation]\label{lem:simulate-time}
Given $A_{\GF{}}$ as in \cref{lem:simulate}, the $k$-\XNFA $B$ of that lemma can be constructed in time $O(|\delta|\cdot\mathrm{poly}(k))$.
\end{lemma}

\begin{proof}
We refer to \Cref{def:coordinate-construction} for the three steps.

Step~(1) contributes $\mathrm{poly}(k)$.
Indeed, note that we can find an $k$th degree irreducible polynomial in time
polynomial in $k$ (see, e.g.,~\cite{Shoup}).

Step~(2) contributes $O(|\delta|\cdot \mathrm{poly}(k))$ (at most $|\delta|$ matrices, each computable in $\mathrm{poly}(k)$ time).
In this computation, reducing modulo $f$ a polynomial of degree at most $2 k - 2$ takes time polynomial in $k$.

Step~(3) has a running time that is linear in the output size, $O(|\delta|\cdot k^2)$. The total is $O(|\delta|\cdot \mathrm{poly}(k))$.
\end{proof}

Next, we decompose the $k$-XNFA we created into $k$ XNFA. 
\XNFAdecomposition*

\begin{proof}
By definition, $A$ accepts $w = a_1 \cdots a_l$ iff the row vector
$v := \lambda \cdot \Delta_{a_1} \cdots \Delta_{a_l} \cdot \gamma \in \mathrm{GF}(2)^{1\times k}$
is non-zero, i.e., iff at least one of its $k$ coordinates is non-zero. The $i$th coordinate of $v$ is
$\lambda \cdot \Delta_{a_1} \cdot\ldots\cdot \Delta_{a_l} \cdot \gamma_i$,
which is exactly the acceptance expression of $A_i$. Hence, $A$ accepts $w$ iff at least one $A_i$ does. The equality of state sets and transition matrices is immediate from the definition of $A_i$.
\end{proof}

\begin{remark}\label{rem:decompose_determinise}
The decomposition of \cref{lem:decompose} does \emph{not} incur a significant complexity blow-up: all automata $A_i$ share the same state set and transition structure and differ only in their accepting states. In fact, if we determinise a $k$-\XNFA, we obtain the same DFA as when we determinise any one of the $A_i$'s, but a state in the DFA is declared final if and only if it is final in at least one of the determinised $A_i$'s. Thus, when determinising a $k$-\XNFA with $n^\prime$ states, we obtain a DFA with at most $2^{n'}$ states.
\end{remark}

As a direct consequence of \cref{lem:main}, we have the following: 

\begin{corollary}
With the assumptions as in \cref{lem:main}, let $\eta\in(0,1)$ and suppose
$k \ge \left\lceil \log_2\left(\nicefrac{|w|}{\eta}\right)\right\rceil$.
If $A$ accepts $w$, then at least one of the constructed XNFA accepts
$w$ with probability at least $1-\eta$.
\end{corollary}

By the reduction presented in \cref{lem:main}, we can also conclude the following: 

\begin{corollary}\label{cor:reduction-general}
Suppose that, for some real $\varepsilon>0$, there is an algorithm that decides
XNFA Acceptance for every XNFA $C$ with transition relation $\delta^{(C)}$ and every
input word $u$ in time $O\bigl((|\delta^{(C)}|\cdot|u|)^{1-\varepsilon}\bigr)$.
Then, for every $\eta\in(0,1)$, the acceptance problem for a polynomially
ambiguous NFA $A$ with transition relation $\delta$ and input word $w$ admits a randomised algorithm with
one-sided error at most $\eta$ and running time
$O\!\bigl(|\delta|\cdot\mathrm{poly}(k) + |w|\,k
        + k\cdot\bigl((k^{2}\cdot|\delta|)\,|w|\bigr)^{1-\varepsilon}\bigr)$,
where $k=\lceil\log_2(|w|/\eta)\rceil$. In particular, for constant $\eta$ this is
$\widetilde{O}\bigl(|\delta| + |w| + (|\delta|\cdot|w|)^{1-\varepsilon}\bigr)$,
where $\widetilde{O}$ hides polylogarithmic factors in $|w|$.
\end{corollary}

\begin{proof}
If $w=\varepsilon$, acceptance can be checked directly from the initial and
final states; thus, assume $w\in\Sigma^+$.
Choose
$k=\left\lceil \log_2\left(\nicefrac{|w|}{\eta}\right)\right\rceil$.
Then $2^k\ge \nicefrac{|w|}{\eta}$, and hence $\nicefrac{|w|}{2^k} \le \eta$. By Theorem~\ref{lem:main},
if the original NFA rejects $w$, then all constructed XNFA reject $w$
with probability $1$. If the original NFA accepts $w$, then at least one
of the constructed XNFA accepts $w$ with probability at least
$1-\eta$.

By Theorem~\ref{lem:main}, each constructed XNFA has at most $k^2\cdot |\delta|$ transitions,
and there are $k$ such XNFA. Running the assumed XNFA Acceptance
algorithm on each of them, therefore, takes time
$\O{k\cdot (k^2\cdot|\delta|\cdot|w|)^{1-\varepsilon}}$.
The construction itself takes time $O(|\delta|\cdot\operatorname{poly}(k))$ by
Theorem~\ref{lem:main}. For constant $\eta$, we have $k=O(\log |w|)$; thus, the first term is
$\widetilde{O}(|\delta|)$ and the second term is
$\widetilde{O}((|\delta|\cdot|w|)^{1-\varepsilon})$.
The term $O(|w| k)$ accounts for preparing the $k$ instances of XNFA
Acceptance, namely for the time needed to write the input word $w$.
\end{proof}